\documentclass[a4paper,12pt,twoside]{article}
\title{\textbf{A decoding algorithm for binary linear codes using Groebner bases}}
\author{Harinaivo ANDRIATAHINY$^{(1)}$\\e-mail : hariandriatahiny@gmail.com\\Jean Jacques Ferdinand RANDRIAMIARAMPANAHY$^{(2)}$\\e-mail : randriamiferdinand@gmail.com\\Toussaint Joseph RABEHERIMANANA$^{(3)}$\\e-mail : rabeherimanana.toussaint@yahoo.fr\\\\$^{1,2,3}$Mention : Mathematics and Computer Science,\\Domain : Sciences and Technologies,\\University of Antananarivo, Madagascar}
\usepackage{geometry}
\usepackage{amsmath,amssymb}
\usepackage{amscd,amsfonts}
\usepackage{amsthm}
\usepackage[english]{babel}
\usepackage[T1]{fontenc}
\usepackage[ansinew]{inputenc}
\swapnumbers\theoremstyle{plain}
\swapnumbers
\theoremstyle{plain}

\usepackage[pdfborder={0 0 0}]{hyperref}
\usepackage{hyperref}

\newtheorem{thm}{Theorem}[section]
\newtheorem{prop}[thm]{Proposition}

\newtheorem{cor}[thm]{Corollary}

\DeclareMathOperator{\lcm}{lcm}

\DeclareMathOperator{\lt}{\mathrm{lt}}

\DeclareMathOperator{\lc}{\mathrm{lc}}
\DeclareMathOperator{\lm}{\mathrm{lm}}

\theoremstyle{definition}
\newtheorem{defi}[thm]{Definition}

\newtheorem{rem}[thm]{Remark}

\DeclareMathOperator{\card}{card}
\DeclareMathOperator{\supp }{supp}

\begin{document}

\maketitle

\begin{abstract}
It has been discovered that linear codes may be described by binomial ideals. This makes it possible to study linear codes by commutative algebra and algebraic geometry methods. In this paper, we give a decoding algorithm for binary linear codes by utilizing the Groebner bases of the associated ideals.
\end{abstract}

\section{Introduction}
Coding theory is important for data transmission through noisy communication channels. During the transmission, errors may occur. Linear codes form an important class of error correcting codes.\\
Bruno Buchberger introduced the theory of Groebner bases for polynomial ideals in $1965$. The Groebner bases theory can be used to solve some problems concerning the ideals by developing computations in multivariate polynomial rings.\\
Connection between linear codes and ideals in polynomial rings was presented in \cite{quintana}. It was proved that a Groebner basis of the ideal associated to a binary linear code can be used for determining the minimum distance . In \cite{GBlinearcode}, it has been proved that a linear code can be described by a binomial ideal, and a Groebner basis with respect to a lexicographic order for the binomial ideal is determined .\\
The aim of this paper is to give full decoding algorithm for binary linear codes via Groebner bases, which completes the decoders presented in \cite{GBlinearcode,quintana}.

\section{Groebner bases}
In this section, we recall some definitions and basic properties about Groebner basis (see\cite{clo}) which are useful to our results.\\
Let $k$ be an arbitrary field. $\mathbb{N}$ denotes the set of non negative integers. A monomial in the $m$ variables $X_1,\dots,X_m$ is a product of the form $X_1^{\alpha_1}\dots X_m^{\alpha_m}$, where all the exponents $\alpha_1,\dots,\alpha_m$ are in $\mathbb{N}$. Let $\alpha = (\alpha_1,\dots,\alpha_m)\in\mathbb{N}^m$. We set $X^\alpha = X_1^{\alpha_1}\dots X_m^{\alpha_m}$. When $\alpha = (0,\dots,0)$, note that $X^\alpha = 1$. We also let $\mid\alpha \mid = \alpha_1+\dots+\alpha_m$ denote the total degree of the monomial $X^\alpha$. We define the sum $\alpha + \beta =(\alpha_1+\beta_1,\dots,\alpha_m+\beta_m)\in \mathbb{N}^m$ with $\beta = (\beta_1,\dots,\beta_m)\in\mathbb{N}^m$.\\
A polynomial $f$ in $X_1,\dots,X_m$ with coefficients in $k$ is a finite linear combination with coefficients in $k$ of monomials. A polynomial $f$ will be written in the form $f = \sum_{\alpha}^{}a_\alpha X^\alpha$, $a_\alpha\in k$, where the sum is over a finite number of m-tuples $\alpha = (\alpha_1,\dots,\alpha_m)$.\\
$k[X_1,\dots,X_m]$ denotes the ring of all polynomials in $X_1,\dots,X_m$ with coefficients in $k$. A monomial order on $k[X_1,\dots,X_m]$ is any relation $>$ on $\mathbb{N}^m$, or equivalently, any relation on the set of monomials $X^\alpha$, $\alpha\in\mathbb{N}^m$, satisfying : 
\begin{itemize}
\item[(i)] $>$ is a total ordering on $\mathbb{N}^m$,
\item[(ii)] if $\alpha>\beta$ and $\gamma\in\mathbb{N}^m$, then $\alpha+\gamma>\beta+\gamma$,
\item[(iii)] $>$ is a well-ordering on $\mathbb{N}^m$.
\end{itemize}
A first example is the lexicographic order. Let $\alpha = (\alpha_1,\dots,\alpha_m)$ and $\beta = (\beta_1,\dots,\beta_m)\in\mathbb{N}^m$. We say $\alpha>_{lex}\beta$ if, in the vector difference $\alpha - \beta\in\mathbb{Z}^m$, the left-most nonzero entry is positive. We will write $X^\alpha>_{lex} X^\beta$ if $\alpha>_{lex}\beta$.\\
A second example is the graded lexicographic order. Let $\alpha, \beta\in \mathbb{N}^m$, we say $\alpha>_{grlex}\beta$ if $\mid\alpha\mid\; >\;\mid\beta\mid$, or $\mid\alpha\mid\;=\;\mid\beta\mid$ and $\alpha>_{lex}\beta$.\\
Let $f=\sum_\alpha a_\alpha X^\alpha$ be a nonzero polynomial in $k[X_1,\dots,X_m]$ and let $>$ be a monomial order. The multidegree of $f$ is $multideg(f)= \max(\alpha\in\mathbb{N}^m/ a_\alpha\neq 0)$, the maximum is taken with respect to $>$. The leading coefficient of $f$ is $\lc(f)=a_{multideg(f)}\in k$. The leading monomial of $f$ is $\lm(f)=X^{multideg(f)}$.  The leading term of $f$ is $\lt(f)=\lc(f).\lm(f)$.
\begin{thm}
Fix a monomial order on $\mathbb{N}^m$, and let $F=(f_1,\dots,f_s)$ be an ordered s-tuple of polynomials in $k[X_1,\dots,X_m]$. Then every $f\in k[X_1,\dots,X_m]$ can be written as $f=a_1f_1+\dots+a_sf_s+r$, where $a_i, r\in k[X_1,\dots,X_m]$, and either $r=0$ or $r$ is a linear combination, with coefficients in $k$, of monomials, none of which is divisible by any of $\lt(f_1),\dots,\lt(f_s)$. We will call $r$ a remainder of $f$ on division by $F$. Furthermore, if $a_if_i\neq 0$, then we have $multideg(f)\geq multideg(a_if_i)$.
\end{thm} 
\begin{rem}\label{rem1}
The operation of computing remainders on division by $F=(f_1,\dots, f_s)$ is linear over $k$. That is, if the remainder on division of $g_i$ by $F$ is $r_i$, $i=1, 2$, then, for any $c_1, c_2\in k$, the remainder on division of $c_1g_1+c_2g_2$ is $c_1r_1+c_2r_2$.
\end{rem}

Let $I\subseteq k[X_1,\dots,X_m]$ be an ideal other than $\{0\}$. We denote by $\lt(I)$ the set of leading terms of elements of $I$. Thus
\begin{center}
$\lt(I)=\{cX^\alpha/ \text{there exists}\  f\in I\ \  \text{with}\ \lt(f)=cX^\alpha\}$.
\end{center}
For each subset $S$ of $k[X_1,\dots,X_m]$, the ideal of $k[X_1,\dots,X_m]$ generated by $S$ is denoted by $\langle S\rangle$.
\begin{thm}[Hilbert Basis Theorem]
Every ideal $I\subseteq k[X_1,\dots,X_m]$ has a finite generating set. That is, $I=\langle g_1,\dots, g_t\rangle$ for some polynomials $g_1,\dots, g_t\in I$.
\end{thm}
\begin{defi}
Fix a monomial order. A finite subset $G=\{g_1,\dots,g_t\}$ of an ideal $I\subseteq k[X_1,\dots,X_m]$ is said to be a Groebner basis for $I$ if \[\langle\lt(g_1),\dots,\lt(g_t)\rangle=\langle\lt(I)\rangle.\]
\end{defi}
\begin{prop}
Fix a monomial order. Every ideal $I$ in the polynomial ring $k[X_1,\dots,X_m]$ other than $\{0\}$ has a Groebner basis. Furthermore, any Groebner basis for an ideal $I$ is a basis of $I$.
\end{prop}
\begin{prop}\label{prop1}
Let $G=\{g_1,\dots,g_t\}$ be a Groebner basis for an ideal $I\subseteq k[X_1,\dots,X_m]$ and let $f\in k[X_1,\dots,X_m]$. Then there is a unique $r\in k[X_1,\dots,X_m]$ with the following properties :\\
$(i)$\ \ No term of $r$ is divisible by any of $\lt(g_1),\dots, \lt(g_t)$.\\
$(ii)$\ \ There is $g\in I$ such that $f=g+r$.\\
In particular, $r$ is the remainder on division of $f$ by $G$ no matter how the elements of $G$ are listed when using the division algorithm. 
\end{prop}
\begin{cor}\label{cor3}
Let $G=\{g_1,\dots,g_t\}$ be a Groebner basis for an ideal $I\subseteq k[X_1,\dots,X_m]$ and let $f\in k[X_1,\dots,X_m]$. Then $f\in I$ if and only if the remainder on division of $f$ by $G$ is zero.
\end{cor}
We will write $\overline{f}^F$ for the remainder on division of $f$ by the ordered s-tuple $F=(f_1,\dots, f_s)$. If $F$ is a Groebner basis for $\langle f_1,\dots, f_s\rangle$, then we can regard $F$ as a set without any particular order.\\ 
Let $f, g\in k[X_1,\dots,X_m] $ be nonzero polynomials. If $multideg(f)=\alpha= (\alpha_1,\dots,\alpha_m)$ and $multideg(g)=\beta= (\beta_1,\dots,\beta_m)$, then let $\gamma=(\gamma_1,\dots, \gamma_m)$ where $\gamma_i=\max(\alpha_i, \beta_i)$ for each $i$. We call $X^\gamma$ the least common multiple of $\lm(f)$ and $\lm(g)$, written $X^\gamma=\lcm(\lm(f),\lm(g))$. The S-polynomial of $f$ and $g$ is the combination 
\[S(f,g)=\dfrac{X^\gamma}{\lt(f)}.f-\dfrac{X^\gamma}{\lt(g)}.g\]
\begin{thm}\label{thmgrob}
Let $I$ be a polynomial ideal. A basis $G=\{g_1,\dots,g_t\}$ for $I$ is a Groebner basis for $I$ if and only if for all pairs $i\neq j$, the remainder on division of $S(g_i,g_j)$ by $G$ listed in some order is zero.
\end{thm}
\begin{rem}
Let $I\subseteq k[X_1,\dots,X_m]$ be an ideal, and let $G$ be a Groebner basis of $I$. Then $\overline{f}^G=\overline{g}^G$ if and only if $f-g\in I$.
\end{rem}
A reduced Groebner basis for a polynomial ideal $I$ is a Groebner basis $G$ for $I$ such that :\\
  $(i)$\ \ $\lc(p)=1$ for all $p\in G$\\
  $(ii)$\ \ For all $p\in G$, no monomial of $p$ lies in $\langle \lt(G-\{p\})\rangle$
\begin{prop}
Let $I\neq \{0\}$ be a polynomial ideal. Then, for a given monomial order, $I$  has a unique reduced Groebner basis.
\end{prop}

\section{Linear codes and binomial ideals}
Let $\mathbb{F}_p$ be the finite field with $p$ elements where $p$ is a prime number. A linear code $\mathcal{C}$ of length $n$ and dimension $k$  over $\mathbb{F}_p$ is the image of a linear (injective) mapping
 \[ \psi:\mathbb{F}_p^k  \longrightarrow\mathbb{F}_p^n\]
where $ k\leq n$. The elements  of $\mathcal{C}$ are called the codewords. Each word $c=(c_1,...,c_n)\in\mathbb{F}_p^n$ may be represented by the monomial $X^c=X_1^{c_1}\ldots X_n^{c_n}$ and is considered as an integral vector in $X^c$. if $c=(0,...,0)$, then $X^c=1$. We define the support of  an element $c=(c_1,...,c_n)\in \mathbb{F}_p^n$ by $\supp(c):= \{i / c_i\neq 0 \}$. The weight of a word  $c=(c_1,...,c_n)\in\mathbb{F}_p^n$ (or $X^c$) is defined by $w(c):= \card(\supp(c))$, i.e. the number of nonzero entries in $c$. The minimum distance of the linear code  $\mathcal{C}$ is $d:= \min\{d(x,y)/ x,y\in\mathcal{C},x\not = y\}$ where $d(x,y):= \card(\{i / x_i\neq y_i \})$ with $x=(x_1,...,x_n)$ and $y=(y_1,...,y_n)$. We have also $d:= \min\{ w(x) \slash x\in\mathcal{C},x\not = 0\}$. A linear code $\mathcal{C}$ of length $n$ and dimension $k$ is called an $[n,k]$-code. Moreover, if the minimum distance is $d$, we say that $\mathcal{C}$ is an $[n,k,d]$-code.\\
 Let $\mathcal{C}$ be an $[n,k]$-code, $e_i=(\zeta_{i1},...,\zeta_{ik})$ where $i=1,...,k$ the canonical basis of  $\mathbb{F}_p^k$ and $\psi(e_i)=(g_{i1},...,g_{in})$ . The generating matrix of $\mathcal{C}$ is the matrix of dimension  $k \times n$ defined by $G=(g_{ij})$ where $g_{ij}\in\mathbb{F}_p$. The linear code  $\mathcal{C}$ is represented as follows $\mathcal{C}=\{xG / \; x\in\mathbb{F}_p^k\}$. We will say that $G$ is in standard form if $G=(I_k\mid M)$ where $I_k$ is the  $k \times k$ identity matrix.\\
 Let $\mathcal{C}$ be an $[n,k]$-code over $\mathbb{F}_p$. Define the ideal associated with $\mathcal{C}$ as (see\cite{quintana,zi2})
\begin{equation}
\label{ideal}
I_{\mathcal{C}} :=\langle X^c-X^{c^\prime} \ \ \mid c-c^\prime\in\mathcal{C} \rangle + \langle X_i^p-1 \mid 1\leq i \leq n\rangle.
\end{equation}
Let $\mathcal{C}$ be an $[n,k]$-code over  $\mathbb{F}_p$ and 
\begin{equation}\label{matricestandard}
G = (g_{ij}) = (I_k\mid M)
\end{equation}
a generating matrix in standard form .
Let $m_i$ be the vector of length $n$ over $\mathbb{F}_p$ defined by
\begin{equation}\label{mi}
m_i = (0,\ldots, 0,p-g_{i,k+1},\ldots, p-g_{i,n})
\end{equation}
for $1\leq i\leq k$. We have $X^{m_i} = X_{k+1}^{p-g_{i,k+1}}\dots X_{n}^{p-g_{i,n}}  = \displaystyle{\prod_{j\in \supp (m_i)}^{}X_j^{p-g_{i,j}}}$. In particular, if $\supp (m_i)=\emptyset$, then $X^{m_i}=1$.
\begin{thm}
Let us take the lexicographic order on  $\mathbb{K}[X_1,\dots,X_n]$ with \\$X_1>X_2>\cdots >X_n$. The code ideal  $I_\mathcal{C}$ has the reduced Groebner basis
\begin{equation}
\label{groebner}
\mathcal{G} = \{ X_i-X^{m_i} / 1\leq i\leq k\}\cup \{ X_i^p-1 / k+1\leq i\leq n\}.
\end{equation}
\end{thm}
\begin{proof}
A proof can be found in \cite{GBlinearcode}.
\end{proof}

\section{The decoding algorithm}
We now present our main results and the decoding algorithm. In what follows, we consider the case $p=2$ and $\mathcal{G}$ denotes the reduced Groebner basis as in (\ref{groebner}) for a binary linear code $\mathcal{C}$.
\begin{thm}\label{thm1}
Let $\mathcal{C}$ be an $[n,k,d]$-code over $\mathbb{F}_2$ and suppose that $\mathcal{C}$ is $t$-error-correcting where $t$ is the maximal integer such that $2t+1\leq d$. Let $v\in(\mathbb{F}_2)^n$ be a received word which contains at most $t$ errors. Then the word given by $\overline{(X^v-1)}^\mathcal{G}$ contains at most $t$ nonzero entries if and only if $(X^v-1)-\overline{(X^v-1)}^\mathcal{G}$ represents the codeword that is closest to the received word and the nonzero coordinates of the error vector are among the last $n-k$ coordinates of $v$.
\end{thm}
\begin{proof}
Suppose that the word given by $\overline{(X^v-1)}^\mathcal{G}$ contains at most $t$ nonzero components. By \cite{GBlinearcode}, $(X^v-1)-\overline{(X^v-1)}^\mathcal{G}$ gives the codeword that is closest to the received word. And it is clear that $\overline{(X^v-1)}^\mathcal{G}$ does not contain the variables $X_1,\dots,X_k$.\\
The converse is clear because $\overline{(X^v-1)}^\mathcal{G}$ represents the error vector, thus $\omega(\overline{(X^v-1)}^\mathcal{G})\leq t$. 
\end{proof}
\begin{cor}
Let $\mathcal{C}$ be an $[n,k,d]$-code over $\mathbb{F}_2$ and suppose that $\mathcal{C}$ is $t$-error-correcting where $t$ is the maximal integer such that $2t+1\leq d$. Let $v\in(\mathbb{F}_2)^n$ be a received word which contains at most $t$ errors. Then the word given by $\overline{(X^v-1)}^\mathcal{G}$ contains more than $t$ nonzero entries if and only if there is at least one nonzero coordinate of the error vector among the first $k$ coordinates of $v$.
\end{cor}
From the above discussion, we have the following algorithm.
\begin{thm}\label{thm4}
Let $\mathcal{C}$ be an $[n,k,d]$-code over $\mathbb{F}_2$ and let $\mathcal{G}$ be the reduced Groebner basis for $\mathcal{C}$ defined as in $(\ref{groebner})$. Suppose that the code $\mathcal{C}$ is $t$-error-correcting where $t$ is the maximal integer such that $2t+1\leq d$. Let $u=(u_1,\dots,u_k,u_{k+1},\dots,u_n)\in(\mathbb{F}_2)^n$ be a received word which contains at most $t$ errors. Then $u$ can be decoded by the following algorithm:\\
Input: $u$, $\mathcal{G}$\\
Output: a codeword $c$ that is closest to $u$\\
BEGIN\\
- Compute $\overline{X^u-1}^\mathcal{G}$.
\begin{itemize}
\item[-] If $\omega(\overline{X^u-1}^\mathcal{G})\leq t$, then the codeword $c$ is given by $(X^u-1)-\overline{X^u-1}^\mathcal{G}$.
\item[-] If $\omega(\overline{X^u-1}^\mathcal{G})>t$, then determine $v\in E=\{(a_1,\dots,a_k,0,\dots,0)\slash a_i\in\{0,1\},\ \ \sum_{i=1}^ka_i\leq t\}$ such that\\ $\omega(\overline{X^u-1-(X^v-1)}^\mathcal{G})\leq t-\omega(v)$, thus $X^u-1-(X^v-1)-\overline{X^u-1-(X^v-1)}^\mathcal{G}$ gives the codeword $c$.
\end{itemize}
END
\end{thm}
In the case of the linear code which is one error correcting, we have a simple decoding algorithm
\begin{cor}\label{cororesult}
Let $\mathcal{C}$ be a binary linear code of length $n$ and dimension $k$. Suppose that $\mathcal{C}$ is one error correcting. Let $u\in(\mathbb{F}_2)^n$ be a received word which contains at most one error. 
\begin{itemize}
\item[-] If $\omega(\overline{X^u-1}^\mathcal{G})\leq 1$, then $(X^u-1)-\overline{X^u-1}^\mathcal{G}$ gives the codeword that is closest to the received word.
\item[-] If $\omega(\overline{X^u-1}^\mathcal{G})>1$, then there exists an integer $i$ ($1\leq i \leq k$) such that $\overline{X^u-1}^\mathcal{G}=\overline{X^{v_i}-1}^\mathcal{G}$ where $v_i=(0,\dots,0,1,0\dots,0)$, the integer $1$ is the i-th coordinate of $v_i$ and $c= u + v_i$ is the codeword.
\end{itemize}
\end{cor}
It is clear that the previous result can be easily generalized to linear codes over $\mathbb{F}_p$.

\section{Examples}
We consider the $[7,4]$-code $\mathcal{C}$ over $\mathbb{F}_2$ where the generator matrix is given by 
\[G=\begin{pmatrix}
1 & 0 & 0 & 0 & 1 & 1 & 1 \\
0 & 1 & 0 & 0 & 0 & 1 & 1 \\
0 & 0 & 1 & 0 & 1 & 0 & 1 \\
0 & 0 & 0 & 1 & 1 & 1 & 0 \\
\end{pmatrix}\]
By considering the lexicographic order on $\mathbb{F}_2[X_1,\dots,X_7]$ with $X_1>X_2>\dots>X_7$, the ideal $I_\mathcal{C}$ (\ref{ideal})  has the Groebner basis $\mathcal{G}$ whose elements are 

$f_1 = X_1-X_5X_6X_7$\hspace{3cm}$f_5 = X_5^2-1$

$f_2 = X_2-X_6X_7$\hspace{3.5cm}$f_6 = X_6^2-1$

$f_3 = X_3-X_5X_7$\hspace{3.5cm}$f_7 = X_7^2-1$

$f_4 = X_4-X_5X_6$\\
Let $u=(1,0,0,1,1,0,0)$ be a received word. We have $X^u=X_1X_4X_5$, by the division of $X^u-1$ by $\mathcal{G}$, we obtain $X^u-1=X_4X_5(f_1)+X_6X_7(f_4)+X_5X_7(f_5)+X_5X_7-1$. Since $\mathcal{C}$ is one error correcting and $\omega(X_5X_7-1)=2>1$, then by Corollary \ref{cororesult} and from the expression of $f_3$, there exists $i=3$ such that $\overline{X^u-1}^\mathcal{G}=\overline{X^{v_3}-1}^\mathcal{G}$ where $v_3=(0,0,1,0,0,0,0)$. Thus the codeword is $c=u+v_3=(1,0,1,1,1,0,0)\in\mathcal{C}$. \\
Let $v=(1,1,0,1,0,1,1)\in(\mathbb{F}_2)^7$ be another received word. Since $X^v-1=X_1X_2X_4X_6X_7-1$ then $X^v-1=X_2X_4X_6X_7(f_1)+X_4X_5(f_2)+X_5X_6X_7(f_4)+X_7-1$. We have $\omega(\overline{X^v-1}^\mathcal{G})=\omega(X_7-1)=1\leq 1$. Then by Corollary \ref{cororesult}, the codeword is $ c= (1,1,0,1,0,1,1)+(0,0,0,0,0,0,1) = (1,1,0,1,0,1,0)\in\mathcal{C}$.

\end{document}